\documentclass[12pt,a4paper]{article}

\usepackage[top=2cm, bottom=2cm, left=2cm, right=2cm]{geometry}

\usepackage{enumitem}
\usepackage{amsthm}
\usepackage{amsfonts}

\newtheorem{theorem}{Theorem}[section]
\newtheorem{lemma}[theorem]{Lemma}
\newtheorem{corollary}[theorem]{Corollary}
\newtheorem{proposition}[theorem]{Proposition}
\newtheorem{definition}[theorem]{Definition}

\usepackage[lined, algoruled, linesnumbered]{algorithm2e}
\usepackage{color,soul}
\definecolor{cblue}{rgb}{0.36, 0.54, 0.66}
\definecolor{cred}{rgb}{1, 0, 0}

\usepackage{graphicx}

\usepackage{hyperref,xcolor}
\definecolor{winered}{rgb}{0.5,0,0}
\hypersetup
{
    pdfborder={0 0 0},
    colorlinks=true,
    linkcolor={winered},
    urlcolor={winered},
    filecolor={winered},
    citecolor={winered},
    linktoc=all,
}

\author{Evangelos Kosinas\thanks{University of Ioannina, Greece.  E-mail: \texttt{ekosinas@cs.uoi.gr}. The research work was supported by the Hellenic Foundation for Research and Innovation (HFRI) under the 3rd Call for HFRI PhD
Fellowships (Fellowship Number: 6547).}}

\begin{document}

\title{Connectivity Queries under Vertex Failures: Not Optimal, but Practical} 

\maketitle

\begin{abstract}
We revisit once more the problem of designing an oracle for answering connectivity queries in undirected graphs in the presence of vertex failures. Specifically, given an undirected graph $G$ with $n$ vertices and $m$ edges and an integer $d_{\star}\ll n$, the goal is to preprocess the graph in order to construct a data structure $\mathcal{D}$ such that, given a set of vertices $F$ with $|F|=d\leq d_{\star}$, we can derive an oracle from $\mathcal{D}$ that can efficiently answer queries of the form ``is $x$ connected with $y$ in $G\setminus F$?''. Very recently, Long and Saranurak (FOCS 2022) provided a solution to this problem that is almost optimal with respect to the preprocessing time, the space usage, the update time, and the query time. However, their solution is highly complicated, and it seems very difficult to be implemented efficiently. Furthermore, it does not settle the complexity of the problem in the regime where $d_{\star}$ is a constant. Here, we provide a much simpler solution to this problem, that uses only textbook data structures. Our algorithm is deterministic, it has preprocessing time and space complexity $O(d_{\star}m\log n)$, update time $O(d^4 \log n)$, and query time $O(d)$. These bounds compare very well with the previous best, especially considering the simplicity of our approach. In fact, if we assume that $d_{\star}$ is a constant ($d_{\star}\geq 4$), then our algorithm provides some trade-offs that improve the state of the art in some respects. Finally, the data structure that we provide is flexible with respect to $d_{\star}$: it can be adapted to increases and decreases, in time and space that are almost proportional to the change in $d_{\star}$ and the size of the graph.
\end{abstract}

\clearpage

\section{Introduction}


In this paper we deal with the following problem. Given an undirected graph $G$ with $n$ vertices and $m$ edges, and a fixed integer $d_{\star}$ ($d_{\star}\ll n$), the goal is to construct a data structure $\mathcal{D}$ that can be used in order to answer connectivity queries in the presence of at most $d_{\star}$ vertex-failures. More precisely, given a set of vertices $F$, with $|F|\leq d_{\star}$, we must be able to efficiently derive an oracle from $\mathcal{D}$, which can efficiently answer queries of the form ``are the vertices $x$ and $y$ connected in $G\setminus F$?''. In this problem, we want to simultaneously optimize the following parameters: $(1)$ the construction time of $\mathcal{D}$ (preprocessing time), $(2)$ the space usage of $\mathcal{D}$, $(3)$ the time to derive the oracle from $\mathcal{D}$ given $F$ (update time), and $(4)$ the time to answer a connectivity query in $G\setminus F$.
This problem is very well motivated; it has attracted the attention of researchers for more than a decade now, and it has many interesting variations. The reader is referred to \cite{DBLP:journals/siamcomp/DuanP20} or \cite{DBLP:conf/focs/LongS22} for the details on the history of this problem and its variations.  

\subsection{Previous work}

Despite being extensively studied, it is only very recently that an almost optimal solution was provided by Long and Saranurak~\cite{DBLP:conf/focs/LongS22}. Specifically, they provided a deterministic algorithm that has $\hat{O}(m)+\tilde{O}(d_{\star}m)$ preprocessing time, uses $O(m\log^{*}n)$ space, and has $\hat{O}(d^2)$ update time and $O(d)$ query time.\footnote{The symbol $\hat{O}$ hides subpolynomial (i.e. $n^{o(1)}$) factors, and $\tilde{O}$ hides polylogarithmic factors. The hidden expressions in the time-bounds are not specified by the authors in their overview. Also, the description for the $\log^{*}n$ function that appears in the space complexity is that it ``can be substituted with any slowly growing function''. One thing that is explicitly stated, however, is that the hidden subpolynomial factors are worse than polylogarithmic. We must emphasize that the difficulty in stating the precise bounds is partly due to there being various trade-offs in the functions involved, and is partly indicative of the complexity of the techniques that are used.} 
This improves on the previous best deterministic solution by Duan and Pettie~\cite{DBLP:journals/siamcomp/DuanP20}, that has $O(mn\log n)$ preprocessing time, uses $O(d_{\star}m\log n)$ space, and has $O(d^3\log^3{n})$ update time and $O(d)$ query time. We note that there are more solutions to this problem, that optimize some parameters while sacrifising others (e.g., in the solution of Pilipczuk et al.~\cite{DBLP:conf/icalp/PilipczukSSTV22}, there is no dependency on $n$ in the update time, but this is superexponential in $d_{\star}$, and the preprocessing time is $O(mn^{2}2^{2^{O(d_{\star})}})$). We refer to Table 1 in reference~\cite{DBLP:conf/focs/LongS22} for more details on the best known (upper) bounds for this problem. We also refer to Theorem 1.2 in \cite{DBLP:conf/focs/LongS22} for a summary of known (conditional) lower bounds, that establish the optimality of~\cite{DBLP:conf/focs/LongS22}.

\subsection{Our contribution}
The bounds that we mentioned are the best known for a deterministic solution. In practice, one would prefer the solution of Long and Saranurak~\cite{DBLP:conf/focs/LongS22}, because that of Duan and Pettie~\cite{DBLP:journals/siamcomp/DuanP20} has preprocessing time $O(mn\log n)$, which can be prohibitively slow for large enough graphs. However, the solution in \cite{DBLP:conf/focs/LongS22} is highly complicated, and it seems very difficult to be implemented efficiently. This is a huge gap between theory and practice. Furthermore, the (hidden) dependence on $n$ in the time-bounds of \cite{DBLP:conf/focs/LongS22} is not necessarily optimal if we assume that $d_{\star}$ is a constant for our problem. 
We note that this is a problem with various parameters, and thus it is very difficult to optimize all of them simultaneously.

Considering that this is a fundamental connectivity problem, we believe that it is important to have a solution that is relatively simple to describe and analyze, compares very well with the best known bounds (even improves them in some respects), opens a new direction to settle the complexity of the problem, and can be readily implemented efficiently.

In this paper, we exhibit a solution that has precisely those characteristics. We present a deterministic algorithm that has preprocessing time $O(d_{\star}m\log n)$, uses space $O(d_{\star}m\log n)$, and has $O(d^4 \log n)$ update time and $O(d)$ query time.\footnote{The $\log$ factors in the space usage and the time for the updates can be improved with the use of more sophisticated 2D-range-emptiness data structures, such as those in \cite{DBLP:conf/compgeom/ChanLP11}.} Our approach is arguably the simplest that has been proposed for this problem. The previous solutions rely on sophisticated tree decompositions of the original graph. Here, instead, we basically rely on a single DFS-tree, and we simply analyze its connected components after the removal of a set of vertices. It turns out that there is enough structure to allow for an efficient solution (see Section~\ref{the idea}).
	
\begin{table*}[h]
	\centering
	\hspace*{-0cm}
\renewcommand{\arraystretch}{1.3}
\begin{tabular}{ |c|c|c|c|c| } 
 \hline
 {} & Preprocessing & Space & Update & Query\\ \hline\hline
 {Pilipczuk et al.~\cite{DBLP:conf/icalp/PilipczukSSTV22}} & $O(2^{2^{O(d_{\star})}}mn^2)$ & $O(2^{2^{O(d_{\star})}}m)$ & $-$ & $O(2^{2^{O(d_{\star})}})$\\ \hline
 {Duan and Pettie~\cite{DBLP:journals/siamcomp/DuanP20}} & $O(mn\log n)$ & $O(d_{\star}m\log n)$ & $O(d^3\log^3{n})$ & $O(d)$ \\ \hline
 {Long and Saranurak~\cite{DBLP:conf/focs/LongS22}} & $\hat{O}(m)+\tilde{O}(d_{\star}m)$ & $O(m\log^{*}{n})$ & $\hat{O}(d^2)$ & $O(d)$\\ \hline
 {\textbf{This paper}} & $O(d_{\star}m\log n)$ & $O(d_{\star}m\log n)$ & $O(d^4\log n)$ & $O(d)$\\ \hline

\end{tabular}

	\vspace{1em}
	
	\caption{Comparison of the best-known deterministic bounds. We note that $m$ can be replaced with $\bar{m}=\mathit{min}\{m,d_{\star}n\}$, using the sparsification of Nagamochi and Ibaraki~\cite{DBLP:journals/algorithmica/NagamochiI92}. The data structure of Pilipczuk et al. does not support an update phase, but answers queries directly, given a set of (at most $d_{\star}$) failed vertices and two query vertices.\label{table:bounds}}
\end{table*}


\begin{table*}[h]
	\centering
	\hspace*{-0cm}
\renewcommand{\arraystretch}{1.3}
\begin{tabular}{ |c|c|c|c|c| } 
 \hline
 {} & Preprocessing & Space & Update & Query\\ \hline\hline
 {Pilipczuk et al.~\cite{DBLP:conf/icalp/PilipczukSSTV22}} & $O(mn^2)$ & $O(m)$ & $-$ & $O(1)$\\ \hline
 {Duan and Pettie~\cite{DBLP:journals/siamcomp/DuanP20}} & $O(mn\log n)$ & $O(m\log n)$ & $O(\log^3{n})$ & $O(1)$ \\ \hline
 {Long and Saranurak~\cite{DBLP:conf/focs/LongS22}} & $\hat{O}(m)+\tilde{O}(m)$ & $O(m\log^{*}{n})$ & $\hat{O}(1)$ & $O(1)$\\ \hline
 {\textbf{This paper}} & $O(m\log n)$ & $O(m\log n)$ & $O(\log n)$ & $O(1)$\\ \hline

\end{tabular}

	\vspace{1em}
	
	\caption{Comparison of the best-known deterministic bounds, when $d_{\star}$ is a fixed (small) constant. Although the algorithm of Pilipczuk et al. has the best space and query-time bounds, it has very large preprocessing time. Our solution has the best preprocessing time, and also better update time compared to the solutions of \cite{DBLP:journals/siamcomp/DuanP20} and \cite{DBLP:conf/focs/LongS22}. Furthermore, our space usage is almost linear.\label{table:bounds2}}
\end{table*}

	
The bounds that we provide compare very well with the previous best, especially considering the simplicity of our approach. (See Tables~\ref{table:bounds} and \ref{table:bounds2}.) 
In fact, as we can see in Table~\ref{table:bounds}, our solution is the best choice for implementations, considering that the algorithm of Long and Saranurak is very difficult to be implemented within the claimed time-bounds. 
Furthermore, if we assume that $d_{\star}$ is a constant ($d_{\star}\geq 4$), then, as we can see in Table~\ref{table:bounds2}, our algorithm provides some trade-offs, that improve the state of the art in some respects. 

Finally, the data structure that we provide is flexible with respect to $d_{\star}$: it can be adapted to increases and decreases, in time and space that are almost proportional to the change in $d_{\star}$ and the size of the graph (see Corollary~\ref{corollary:resize}). We do not know if any of the previous solutions has this property. It is a natural question whether we can efficiently update the data structure so that it can handle more failures (or less, and thereby free some space). As far as we know, we are the first to take notice of this aspect of the problem. 

\section{Preliminaries}
We assume that the reader is familiar with standard graph-theoretical terminology (see, e.g., \cite{DBLP:books/daglib/0030488}). The notation that we use is also standard. Since we deal with connectivity under \emph{vertex} failures, it is sufficient to consider simple graphs as input to our problem (because the existence of parallel edges does not affect the connectivity relation). However, during the update phase, we construct a multigraph that represents the connectivity relationship between some connected components after removing the failed vertices (Definition~\ref{definition:R}). The parallel edges in this graph are redundant, but they may be introduced by the algorithm that we use, and it would be costly to check for redundancy throughout.

It is also sufficient to assume that the input graph $G$ is connected. Because, otherwise, we can initialize a data structure on every connected component of $G$; the updates, for a given set of failures, are distributed to the data structures on the connected components, and the queries for pairs of vertices that lie in different connected components of $G$ are always $\mathit{false}$.
We use $G$ to denote the input graph throughout; $n$ and $m$ denote its number of vertices and edges, respectively.
For any two integers $x,y$, we use the interval notation $[x,y]$ to denote the set $\{x,x+1,\dots,y\}$. (If $x>y$, then $[x,y]=\emptyset$.)

\subsection{DFS-based concepts}
Let $T$ be a DFS-tree of $G$, with start vertex $r$ \cite{DBLP:journals/siamcomp/Tarjan72}. We use $p(v)$ to denote the parent of every vertex $v\neq r$ in $T$ ($v$ is a child of $p(v)$). For any two vertices $u,v$, we let $T[u,v]$ denote the simple tree path from $u$ to $v$ on $T$.  For every two vertices $u$ and $v$, if the tree path $T[r,u]$ uses $v$, then we say that $v$ is an ancestor of $u$ (equivalently, $u$ is a descendant of $v$). In particular, a vertex is considered to be an ancestor (and also a descendant) of itself. It is very useful to identify the vertices with their order of visit during the DFS, starting with $r\leftarrow 1$. Thus, if $v$ is an ancestor of $u$, we have $v<u$. For any vertex $v$, we let $T(v)$ denote the subtree rooted at $v$, and we let $\mathit{ND}(v)$ denote the number of descendants of $v$ (i.e., $\mathit{ND}(v)=|T(v)|$). Thus, we have that $T(v)=[v,v+\mathit{ND}(v)-1]$, and therefore we can check the ancestry relation in constant time. Two children $c$ and $c'$ of a vertex $v$ are called \emph{consecutive children} of $v$ (in this order), if $c'$ is the minimum child of $v$ with $c'>c$. Notice that, in this case, we have $T(c)\cup T(c')=[c,c'+\mathit{ND}(c')-1]$.    

A DFS-tree $T$ has the following extremely convenient property: the endpoints of every non-tree edge of $G$ are related as ancestor and descendant on $T$ \cite{DBLP:journals/siamcomp/Tarjan72}, and so we call those edges \emph{back-edges}. Our whole approach is basically an exploitation of this property, which does not hold in general rooted spanning trees of $G$ (unless they are derived from a DFS traversal, and only then \cite{DBLP:journals/siamcomp/Tarjan72}). To see why this is relevant for our purposes, consider what happens when we remove a vertex $f\neq r$ from $T$. Let $c_1,\dots,c_k$ be the children of $f$ in $T$. Then, the connected components of $T\setminus f$ are given by $T(c_1),\dots,T(c_k)$ and $T(r)\setminus T(f)$. A subtree $T(c_i)$, $i\in\{1,\dots,k\}$, is connected with the rest of the graph in $G\setminus f$ if and only if there is a back-edge that stems from $T(c_i)$ and ends in a proper ancestor of $f$. Now, this problem has an algorithmically elegant solution. Suppose that we have computed, for every vertex $v\neq r$, the \emph{lowest} proper ancestor of $v$ that is connected with $T(v)$ through a back-edge. We denote this vertex as $\mathit{low}(v)$. 
Then, we may simply check whether $\mathit{low}(c_i)<f$, in order to determine whether $T(c_i)$ is connected with $T(r)\setminus T(f)$ in $G\setminus f$. 

We extend the concept of the $\mathit{low}$ points, by introducing the $\mathit{low}_k$ points, for any $k\in\mathbb{N}$. These are defined recursively, for any vertex $v\neq r$, as follows. $\mathit{low}_1(v)$ coincides with $\mathit{low}(v)$. Then, supposing that we have defined $\mathit{low}_k(v)$ for some $k\in\mathbb{N}$, we define $\mathit{low}_{k+1}(v)$ as $\mathit{min}(\{y\mid \exists\mbox{ a back-edge } (x,y) \mbox{ such that } x\in T(v) \mbox{ and } y<v\}\setminus\{\mathit{low}_1(v),\dots,\mathit{low}_k(v)\})$. Notice that $\mathit{low}_k(v)$ may not exist for some $k\in\mathbb{N}$ (and this implies that $\mathit{low}_{k'}(v)$ does not exist, for any $k'>k$). If, however, $\mathit{low}_k(v)$ exists, then $\mathit{low}_{k'}(v)$, for any $k'<k$, also exists, and we have $\mathit{low}_1(v)<\mathit{low}_2(v)<\dots<\mathit{low}_k(v)$. Notice that the existence of $\mathit{low}_k(v)$ implies that there is a back-edge $(x,\mathit{low}_k(v))$, where $x$ is a descendant of $v$.

\begin{proposition}
\label{proposition:low_points}
Let $T$ be a DFS-tree of a simple graph $G$, and assume that the adjacency list of every vertex of $G$ is sorted in increasing order w.r.t. the DFS numbering. Suppose also that, for some $k\in\{0,\dots,n-1\}$, we have computed the $\mathit{low}_1,\dots,\mathit{low}_k$ points of all vertices (w.r.t. $T$), and the set $\{\mathit{low}_1(v),\dots,\mathit{low}_k(v)\}$ is stored in an increasingly sorted array for every $v\neq r$. Then we can compute the $\mathit{low}_{k+1}$ points of all vertices in $O(n\log(k+1))$ time.\footnote{We make the convention that $\log(1)=1$, so that the time to compute the $\mathit{low}_1$ points is $O(n)$.}
\end{proposition}
\begin{proof}
For every $v\neq r$, let $\mathit{lowArray}(v)$ be the array that contains $\{\mathit{low}_1(v),\dots,\mathit{low}_k(v)\}$ in increasing order, plus one more entry which is $\mathit{null}$. Now we process the vertices in a bottom-up fashion (e.g., in reverse DFS order). We will make sure that, when we start processing a vertex, the $\mathit{low}_1,\dots,\mathit{low}_{k+1}$ points of its children are correctly computed $(*)$. 

The processing of a vertex $v\neq r$ is done as follows. First, we perform a binary search within the first $k+1$ entries of the adjacency list of $v$, in order to find the smallest vertex that is greater than $\mathit{low}_k(v)$; if it exists, we insert it in the $k+1$ entry of $\mathit{lowArray}(v)$. Now we process the children of $v$. For every child $c$ of $v$, if the $k+1$ entry of $\mathit{lowArray}(v)$ is $\mathit{null}$, then we perform a binary search in $\mathit{lowArray}(c)$, in order to find the smallest vertex that is greater than $\mathit{low}_k(v)$ and lower than $v$. If it exists, then we insert it in the $k+1$ entry of $\mathit{lowArray}(v)$. Otherwise, if the $k+1$ entry of $\mathit{lowArray}(v)$ is not $\mathit{null}$, then we perform a binary search in $\mathit{lowArray}(c)$, in order to find the smallest vertex $y$ that is greater than $\mathit{low}_k(v)$ and lower than the $k+1$ entry of $\mathit{lowArray}(v)$. If it exists, then we replace the vertex at the $k+1$ entry of $\mathit{lowArray}(v)$ with $y$. Notice that, for the processing of $v$, we need $O((1+\mathit{nChilden}_v)\log (k+1))$ time, where $\mathit{nChilden}_v$ is the number of children of $v$. Thus, the whole algorithm takes $O(n\log(k+1))$ time in total.

Now we have to argue about the correctness of this procedure. Suppose that $(*)$ is true for a vertex $v$ right when we start processing it. (If $v$ is a leaf, then $(*)$ is trivially true.) Let us also suppose that $\mathit{low}_k(v)$ is exists, because otherwise $\mathit{low}_{k+1}(v)$ does not exist and we are done. Consider the segment $y_1,\dots,y_{k+1}$ of the first $k+1$ entries of the adjacency list of $v$. Then, notice that $\mathit{low}_{k+1}(v)\leq y_{k+1}$ (where we let this inequality be trivially true if $y_{k+1}$ is $\mathit{null}$). This is because $\mathit{low}_1(v)$ is at least as low as $y_1$, therefore $\mathit{low}_2(v)$ is at least as low as $y_2$, and so on. Thus, if $\mathit{low}_{k+1}(v)$ exists in the adjacency list of $v$, it coincides with the lowest among $y_1,\dots,y_{k+1}$ that is greater than $\mathit{low}_k(v)$. Otherwise, after the search in the adjacency list of $v$, we just have that the $k+1$ entry of $\mathit{lowArray}(v)$ (if it is not $\mathit{null}$) contains a vertex that is greater than $\mathit{low}_{k+1}(v)$.
Now we check the $\mathit{low}_i$ points of the children of $v$, for $i\in\{1,\dots,k+1\}$. (By $(*)$, these are correctly computed, and they are stored in the $\mathit{lowArray}$ arrays.) First, we notice, as previously, that $\mathit{low}_{k+1}(v)$ is at least as low as the $k+1$ entry in $\mathit{lowArray}(c)$, for any child $c$ of $v$. Thus, if there is a child $c$ of $v$ such that $\mathit{lowArray}(c)$ contains $\mathit{low}_{k+1}(v)$, then this is precisely the smallest vertex in $\mathit{lowArray}(c)$ that is greater than $\mathit{low}_k(v)$, and we correctly insert it in the $k+1$ entry of $\mathit{lowArray}(v)$.

We conclude that, when we finish processing $v$, either the $k+1$ entry of $\mathit{lowArray}(v)$ is $\mathit{null}$ (from which we infer that $\mathit{low}_{k+1}(v)$ does not exist), or it contains a vertex that is greater than $\mathit{low}_k(v)$, but at least as low as any of the first $k+1$ entries of the adjacency list of $v$ that are greater than $\mathit{low}_k(v)$, or the first $k+1$ $\mathit{low}$ points of any of its children that are greater than $\mathit{low}_k(v)$. Thus, $\mathit{low}_{k+1}(v)$ has been correctly computed in the $k+1$ entry of $\mathit{lowArray}(v)$. 
\end{proof}

\begin{corollary}
\label{corollary:low_points}
For any $k\in\{1,\dots,n-1\}$, the $\mathit{low}_1,\dots,\mathit{low}_k$ points of all vertices can be computed in $O(m+kn\log k)$ time. 
\end{corollary}
\begin{proof}
An immediate appplication of Proposition~\ref{proposition:low_points}: we first sort the adjacency lists of all vertices with bucket-sort, and then we just compute the $\mathit{low}_1,\dots,\mathit{low}_k$ points, for all vertices, in this order. This will take time $O(m+n)+O(n\log{1}+n\log{2}+\dots+n\log{k})=O(m+kn\log k)$.
\end{proof}

\section{The algorithm for vertex failures}

\subsection{Initializing the data structure}
We will need the following ingredients in order to be able to handle at most $d_{\star}$ failed vertices.

\begin{enumerate}[label={(\roman*)}]
\item{A DFS-tree $T$ of $G$ rooted at a vertex $r$. The values $\mathit{ND}$ and $\mathit{depth}$ (w.r.t. $T$) must be computed for all vertices. We identify the vertices of $G$ with the DFS numbering of $T$.}
\item{A level-ancestor data structure on $T$.}
\item{A 2D-range-emptiness data structure on the set of the back-edges of $G$ w.r.t. $T$.}
\item{The $\mathit{low}_i$ points of all vertices, for every $i\in\{1,\dots,d_{\star}\}$.}
\item{For every $i\in\{1,\dots,d_{\star}\}$, a DFS-tree $T_i$ of $T$ rooted at $r$, where the adjacency lists of the vertices are given by their children lists sorted in increasing order w.r.t. the $\mathit{low}_i$ point.}
\item{For every $i\in\{1,\dots,d_{\star}\}$, a 2D-range-emptiness data structure on the set of the back-edges of $G$ w.r.t. $T_i$.}
\end{enumerate} 

The $\mathit{depth}$ value in $(i)$ refers to the depths of the vertices in $T$. This is defined for every vertex $v$ as the size of the tree path $T[r,v]$. (Thus, e.g., $\mathit{depth}(r)=1$.) It takes $O(n)$ additional time to compute the $\mathit{depth}$ values during the DFS.

The level-ancestor data structure in $(ii)$ is used in order to answer queries of the form $\mathtt{QueryLA}(v,\delta)\equiv$ ``return the ancestor of $v$ that lies at depth $\delta$''. We use those queries in order to find the children of vertices that are ancestors of other vertices. (I.e., given that $u$ is a descendant of $v$, we want to know the child of $v$ that is an ancestor of $u$.) For our purposes, it is sufficient to use the solution in Section 3 of \cite{DBLP:journals/tcs/BenderF04}, that preprocesses $T$ in $O(n\log n)$ time so that it can answer level-ancestor queries in (worst-case) $O(1)$ time.

The 2D-range-emptiness data structure in $(iii)$ is used in order to answer queries of the form $\mathtt{2D\_range}([X_1,X_2]\times[Y_1,Y_2])\equiv$ ``is there a back-edge $(x,y)$ with $x\in[X_1,X_2]$ and $y\in[Y_1,Y_2]$?''.\footnote{The input to $\mathtt{2D\_range}$ is just the endpoints $X_1,X_2,Y_1,Y_2$ of the query rectangle; we use brackets around them, and the symbol $\times$, just for readability.} We can use a standard implementation for this data structure, that has $O(m\log n)$ space and preprocessing time complexity, and can answer a query in (worst-case) $O(\log n)$ time (see, e.g., Section 5.6 in \cite{DBLP:books/lib/BergCKO08}). The $m$ factor here is unavoidable, because the number of back-edges can be as large as $m-n+1$. However, we note that we can improve the $\log n$ factor in the space and the query time if we use a more sophisticated solution, such as \cite{DBLP:conf/compgeom/ChanLP11}.

The $\mathit{low}_1,\dots,\mathit{low}_{d_{\star}}$ points of all vertices can be computed in $O(m+d_{\star}n\log{d_{\star}})=O(m+d_{\star}n\log n)$ time (Corollary~\ref{corollary:low_points}). We obviously need $O(d_{\star}n)$ space to store them.

For $(v)$, we just perform $d_{\star}$ DFS's on $T$, starting from $r$, where each time we use a different arrangement of the children lists of $T$ as adjacency lists. This takes $O(d_{\star}n)$ time in total, but we do not need to actually store the trees. (In fact, the parent pointer is the same for all of them.) 
What we actually need here is the DFS numbering of the $i$-th DFS traversal, for every $i\in\{1,\dots,d_{\star}\}$, which we denote as $\mathit{DFS}_i$. We keep those DFS numberings stored, and so we need $O(d_{\star}n)$ additional space. The usefulness of performing all those DFS's will become clear in Section~\ref{section:update}. Right now, we only need to mention that, for every $i\in\{1,\dots,d_{\star}\}$, the ancestry relation in $T_i$ is the same as that in $T$. Thus, the $\mathit{low}_1,\dots,\mathit{low}_{d_{\star}}$ points for all vertices w.r.t. $T_i$ are the same as those w.r.t. $T$.

The 2D-range-emptiness data structures in $(vi)$ are used in order to answer queries of the form $\mathtt{2D\_range\_i}([X_1,X_2]\times[Y_1,Y_2])\equiv$ ``is there a back-edge $(x,y)$ with $x\in[X_1,X_2]$ and $y\in[Y_1,Y_2]$?'', where the endpoints of the query rectangle refer to the $\mathit{DFS}_i$ numbering, for $i\in\{1,\dots,d_{\star}\}$. Since the ancestry relation is the same for $T_i$ and $T$, we have that the queries $\mathtt{2D\_range}([X_1,X_2]\times[Y_1,Y_2])$ and $\mathtt{2D\_range\_i}([X_1,X_2]_i\times[Y_1,Y_2]_i)$ are equivalent, where the $i$ index below the brackets means that we have translated the endpoints in the $\mathit{DFS}_i$ numbering.

The construction of the 2D-range-emptiness data structures w.r.t. the DFS-trees $T_1,\dots,T_{d_{\star}}$ takes $O(d_{\star}m\log n)$ time in total. In order to keed those data structures stored, we need $O(d_{\star}m\log n)$ space. Thus, the construction and the storage of the 2D-range-emptiness data structures dominate the space-time complexity overall.

It is easy to see that the list of data structures from $(i)$ to $(vi)$ is flexible w.r.t. $d_{\star}$.
Thus, if $d_{\star}$ increases by $1$, then we need to additionally compute the $\mathit{low}_{d_{\star}+1}$ points of all vertices, the $T_{d_{\star}+1}$ DFS-tree, and the corresponding 2D-range-emptiness data structure. Computing the $\mathit{low}_{d_{\star}+1}$ points takes $O(n\log(d_{\star}+1))=O(n\log n)$ time, and demands an additional $O(n)$ space, assuming that we have sorted the adjacency lists of $G$ in increasing order, and that we have stored the $\mathit{low}_1,\dots,\mathit{low}_{d_{\star}}$ points, for every vertex, in an increasingly sorted array (see Proposition~\ref{proposition:low_points}).

\begin{corollary}
\label{corollary:resize}
Suppose that we have initialized our data structure for some $d_{\star}$, and we want to get a data structure for $d_{\star}+k$. Then we can achieve this in $O(km\log n)$ time, using extra $O(km\log n)$ space.
\end{corollary}

If $d_{\star}$ decreases by $k$, then we just have to discard the $\mathit{low}_{d_{\star}-k+1},\dots,\mathit{low}_{d_{\star}}$ points, the $T_{d_{\star}-k+1},\dots,T_{d_{\star}}$ DFS-trees, and the corresponding 2D-range-emptiness data structures. This will free $O(km\log n)$ space.


\subsection{The general idea}
\label{the idea}
Let $F$ be a set of failed vertices. Then $T\setminus F$ may consist of several connected components, all of which are subtrees of $T$. It will be necessary to distinguish two types of connected components of $T\setminus F$. Let $C$ be a connected component of $T\setminus F$. If no vertex in $F$ is a descendant of $C$, then $C$ is called a \emph{hanging subtree} of $T\setminus F$. Otherwise, $C$ is called an \emph{internal component} of $T\setminus F$. (See Figure~\ref{figure:failed_vertices} for an illustration.) Observe that, while the number of connected components of $T\setminus F$ may be as large as $n-1$ (even if $|F|=1$), the number of internal components of $T\setminus F$ is at most $|F|$. This is an important observation, that allows us to reduce the connectivity of $G\setminus F$ to the connectivity of the internal components.

\begin{figure}[t!]\centering
\includegraphics[trim={0 15cm 0 0}, clip=true, width=0.7\linewidth]{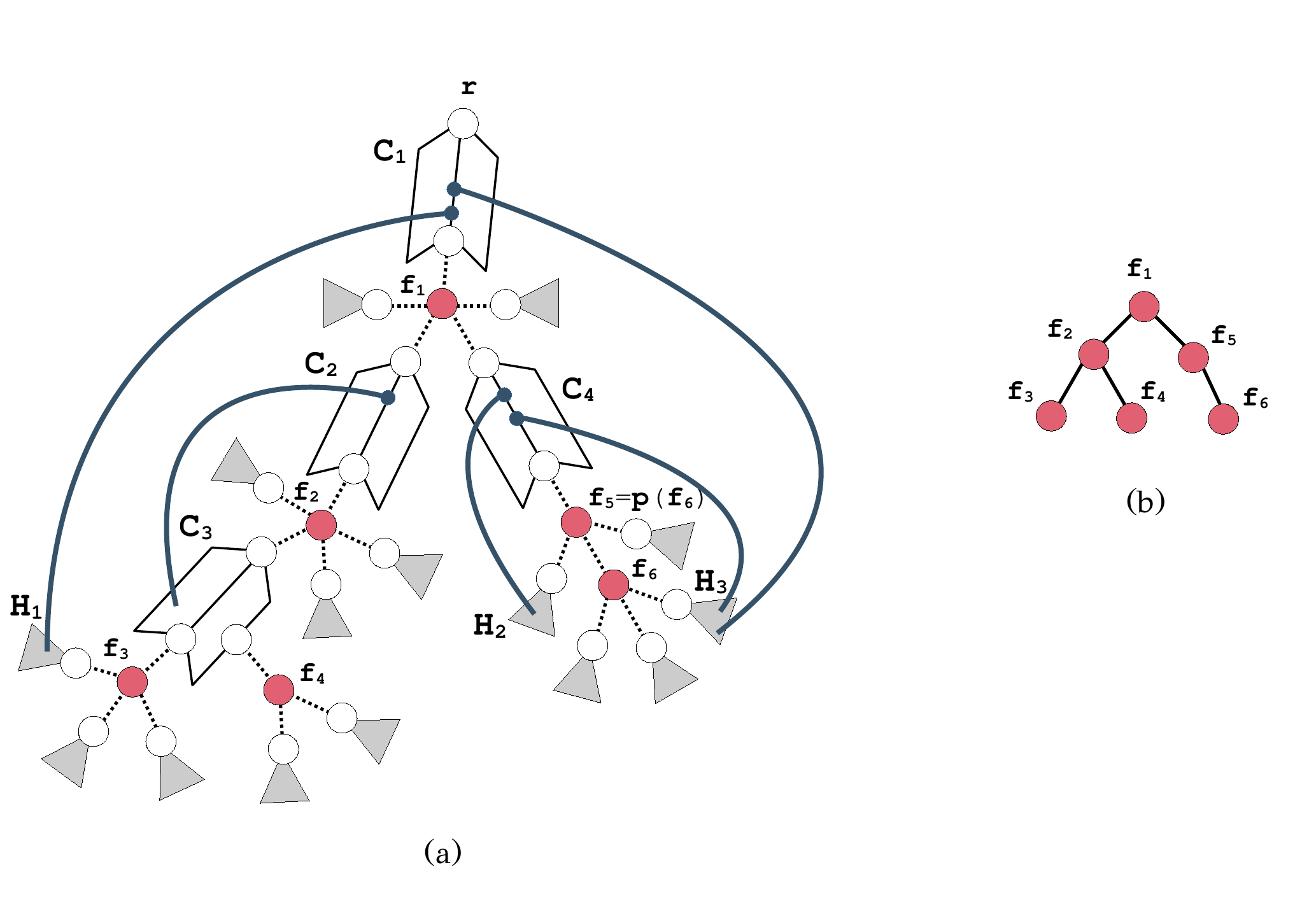}
\caption{\small{(a) A set of failed vertices $F=\{f_1,\dots,f_6\}$ on a DFS-tree $T$, and (b) the corresponding F-forest, which shows the $\mathit{parent}_F$ relation between failed vertices. Notice that $T\setminus F$ is split into several connected components, but there are only four internal components, $C_1$, $C_2$, $C_3$ and $C_4$. The hanging subtrees of $T\setminus F$ are shown with gray color (e.g., $H_1$, $H_2$ and $H_3$). The internal components $C_2$ and $C_3$ remain connected in $G\setminus F$ through a back-edge that connects them directly. $C_1$ and $C_4$ remain connected through the hanging subtree $H_3$ of $f_6$. We have $\partial(C_1)=\{f_1\}$, $\partial(C_2)=\{f_2\}$, $\partial(C_3)=\{f_3,f_4\}$ and $\partial(C_4)=\{f_5\}$. Notice that $f_6$ is the only failed vertex that is not a boundary vertex of an internal component, and it has $\mathit{parent}_F(f_6)=p(f_6)$.}}\label{figure:failed_vertices}
\end{figure}

More precisely, we can already provide a high level description of our strategy for answering connectivity queries between pairs of vertices. Let $x,y$ be two vertices of $G\setminus F$. Suppose first that $x$ belongs to an internal component $C_1$ and $y$ belongs to an internal component $C_2$. Then it is sufficient to know whether $C_1$ and $C_2$ are connected in $G\setminus F$. Otherwise, if either $x$ or $y$ lies in a hanging subtree $C$, then we can substitute $C$ with any internal component that is connected with $C$ in $G\setminus F$. If no such internal component exists, then $x$ and $y$ are connected in $G\setminus F$ if and only if they lie in the same hanging subtree.

Thus, after the deletion of $F$ from $G$, it is sufficient to make provisions so as to be able to efficiently answer the following:

\begin{enumerate}[label={(\arabic*)}]
\item{Given a vertex $x$, determine the connected component of $T\setminus F$ that contains $x$.}
\item{Given two internal components $C_1$ and $C_2$ of $T\setminus F$, determine whether $C_1$ and $C_2$ are connected in $G\setminus F$.}
\item{Given a hanging subtree $C$ of $T\setminus F$, find an internal component of $T\setminus F$ that is connected with $C$ in $G\setminus F$, or report that no such internal component exists.}
\end{enumerate} 

Actually, the most difficult task, and the only one that we provide a preprocessing for (during the update phase), is $(2)$. We explain how to perform $(1)$ and $(3)$ during the process of answering a query, in Section~\ref{section:query}. 
An efficient solution for $(2)$ is provided in Section~\ref{section:update}. 

The general idea is that, since there are at most $d=|F|$ internal components of $T\setminus F$, we can construct a graph with $O(d)$ nodes, representing the internal components of $T\setminus F$, that captures the connectivity relation among them in $G\setminus F$ (see Lemma~\ref{lemma:maintaining_components}). This is basically done with the introduction of some artificial edges between the (representatives of the) internal components. In the following subsection, we state some lemmata concerning the structure of the internal components, and their connectivity relationship in $G\setminus F$.

\subsection{The structure of the internal components}

We will use the roots of the connected components of $T\setminus F$ (viewed as rooted subtrees of $T$) as representantives of them. Now we introduce some terminology and notation. If $C$ is a connected component of $T\setminus F$, we denote its root as $r_C$. If $C$ is a hanging subtree of $T\setminus F$, then $p(r_C)=f$ is a failed vertex, and we say that $C$ is \emph{a hanging subtree of} $f$. If $C,C'$ are two distinct connected components of $T\setminus F$ such that $r_{C'}$ is an ancestor of $r_C$, then we say that $C'$ is an ancestor of $C$. Furthermore, if $v$ is a vertex not in $C$ such that $v$ is an ancestor (resp., a descendant) of $r_C$, then we say that $v$ is an ancestor (resp., a descendant) of $C$. If $C$ is an internal component of $T\setminus F$ and $f$ is a failed vertex such that $p(f)\in C$, then we say that $f$ is a boundary vertex of $C$. The collection of all boundary vertices of $C$ is denoted as $\partial(C)$. Notice that any vertex $b\in\partial(C)$ has the property that there is no failed vertex on the tree path $T[p(b),r_C]$. Conversely, a failed vertex $b$ such that there is no failed vertex on the tree path $T[p(b),r_C]$ is a boundary vertex of $C$. Thus, if $b_1,\dots,b_k$ is the collection of all the boundary vertices of $C$, then $C=T(r_C)\setminus(T(b_1)\cup\dots\cup T(b_k))$.

The following lemma is a collection of properties that are satisfied by the internal components.

\begin{lemma}
\label{lemma:ic-info}
Let $C$ be an internal component of $T\setminus F$. Then:
\begin{enumerate}[label={(\arabic*)}]
\item{Either $r_C=r$, or $p(r_C)\in F$.}
\item{For every vertex $v$ that is a descendant of $C$, there is a unique boundary vertex of $C$ that is an ancestor of $v$.}
\item{Let $f_1,\dots,f_k$ be the boundary vertices of $C$, sorted in increasing order. Then $C$ is the union of the following subsets of consecutive vertices: $[r_C,f_1-1],[f_1+\mathit{ND}(f_1),f_2-1],\dots,[f_{k-1}+\mathit{ND}(f_{k-1}),f_k-1],[f_k+\mathit{ND}(f_k),r_C+\mathit{ND}(r_C)-1]$. (We note that some of those sets may be empty.)}
\end{enumerate}
\end{lemma}
\begin{proof}
$(1)$ If $r_C\neq r$, then $p(r_C)$ is defined. Since $r_C$ is the root of a connected component of $T\setminus F$, we have that $r_C\notin F$. If $p(r_C)\notin F$, then $r_C$ is connected with $p(r_C)$ in $T\setminus F$ through the parent edge $(r_C,p(r_C))$, contradicting the fact that $r_C$ is the root of a connected component of $T\setminus F$. Thus, $p(r_C)\in F$.

$(2)$ Since $v$ is a descendant of $C$, we have that $v\notin C$ and $v$ is a descendant of $r_C$. Since $v$ is a descendant of $r_C$, we have that all vertices on the tree path $T[v,r_C]$ are ancestors of $v$. (Notice that only a vertex on $T[v,r_C]$ can be both an ancestor of $v$ and a boundary vertex of $C$, because all other ancestors of $v$ are lower than $r_C$.) Since $v\notin C$, there must exist at least one failed vertex on $T[v,r_C]$. Take the lowest such failed vertex $b$. Then we have that none of the vertices on the tree path $T[p(b),r_C]$ is a failed vertex, and so $p(b)$ is connected with $r_C$ in $T\setminus F$, and therefore $b$ is a boundary vertex of $C$. 

Now let us suppose, for the sake of contradiction, that there is another vertex $b'$ on $T[v,r_C]$ that is a boundary vertex of $C$. Since $b$ is the lowest with this property, we have that $b'$ is a proper descendant of $b$. Since $b'\in\partial(C)$, we have that there cannot be a failed vertex on the tree path $T[p(b'),r_C]$, contradicting the fact that $b\in T[p(b'),r_C]$. Thus, we have that $b$ is the unique vertex in $\partial(C)$ that is an ancestor of $v$.

$(3)$ The subtree rooted at $r_C$ consists of the vertices in $[r_C,r_C+\mathit{ND}(r_C)-1]$. Since $f_1,\dots,f_k$ are the boundary vertices of $C$, we have that $C=T(r_C)\setminus(T(f_1)\cup\dots\cup T(f_k))$. Therefore, $C=[r_C,r_C+\mathit{ND}(r_C)-1]\setminus([f_1,f_1+\mathit{ND}(f_1)-1]\cup\dots\cup [f_k,f_k+\mathit{ND}(f_k)-1])$. Thus, since $f_1,\dots,f_k$  are sorted in increasing order, we have $C=[r_C,f_1-1]\cup[f_1+\mathit{ND}(f_1),f_2-1]\cup\dots\cup[f_{k-1}+\mathit{ND}(f_{k-1}),f_k-1]\cup[f_k+\mathit{ND}(f_k),r_C+\mathit{ND}(r_C)-1]$.
\end{proof}

We represent the ancestry relation between failed vertices using a forest which we call the \emph{failed vertex forest} (\emph{F-forest}, for short). The F-forest consists of the following two elements. First, for every failed vertex $f$, there is a pointer $\mathit{parent}_F(f)$ to the nearest ancestor of $f$ (in $T$) that is also a failed vertex. If there is no ancestor of $f$ that is a failed vertex, then we let $\mathit{parent}_F(f)=\bot$. And second, every failed vertex $f$ has a pointer to its list of children in the F-forest.

The F-forest can be easily constructed in $O(d^2)$ time: we just have to find, for every failed vertex $f$, the maximum failed vertex $f'$ that is a proper ancestor of $f$; then we set $\mathit{parent}_F(f)=f'$, and we append $f$ to the list of the children of $f'$ in the F-forest.

The next lemma shows how we can check in constant time whether a failed vertex belongs to the boundary of an internal component, and how to retrieve the root of this component.

\begin{lemma}
\label{lemma:f_is_boundary}
A failed vertex $f$ is a boundary vertex of an internal component if and only if $\mathit{parent}_F(f)\neq p(f)$. Now let $f$ be a boundary vertex of an internal component $C$. Then, if $\mathit{parent}_F(f)$ exists, we have that the root of $C$ is the child of $\mathit{parent}_F(f)$ that is an ancestor of $f$. Otherwise, the root of $C$ is $r$. 
\end{lemma}
\begin{proof}
Let $C$ be an internal component such that $f\in\partial(C)$. Then there is no failed vertex on the tree path $T[p(f),r_C]$. In particular, $p(f)\neq\mathit{parent}_F(f)$. Conversely, suppose that $\mathit{parent}_F(f)\neq p(f)$. (We can reject the case $f=r$, because then none of the expressions $\mathit{parent}_F(f),p(f)$ is defined.) If $\mathit{parent}_F(f)$ is not defined, then there is no failed vertex on the tree path $T[p(f),r]$ (i.e., on the path of the ancestors of $f$), and therefore $f$ is a boundary vertex of the internal component with root $r$. Otherwise, if $\mathit{parent}_F(f)$ is defined, then we have that $p(f)$ cannot be a failed vertex (because otherwise we would have $\mathit{parent}_F(f)=p(f)$, because $\mathit{parent}_F(f)$ is the nearest proper ancestor of $f$ that is a failed vertex). Thus, $p(f)$ belongs to a connected component of $T\setminus F$, to which $f$ is a boundary vertex.

Now let $f$ be a boundary vertex of an internal component $C$. This means that there is no failed vertex on the tree path $T[p(f),r_C]$. If $\mathit{parent}_F(f)$ exists, then it must be a proper ancestor of $r_C$. Thus, $r_C\neq r$, and therefore, by Lemma~\ref{lemma:ic-info}$(1)$, we have that $p(r_C)$ is a failed vertex. Since $\mathit{parent}_F(f)$ is the nearest ancestor of $f$ that is a failed vertex, we thus have that $\mathit{parent}_F(f)=p(r_C)$, and therefore $r_C$ is the child of $\mathit{parent}_F(f)$ that is an ancestor of $f$. Otherwise, if $\mathit{parent}_F(f)$ does not exist, this implies that there is no failed vertex on the tree path $T[p(f),r]$. Thus, $f$ is a boundary vertex of the internal component with root $r$.
\end{proof}

Thus, according to Lemma~\ref{lemma:f_is_boundary}, if $f$ is a boundary vertex of an internal component $C$ with $r_C\neq r$, we can retrieve $r_C$ in constant time using a level-ancestor query: i.e., we ask for the ancestor of $f$ (in $T$) whose depth equals that of $\mathit{parent}_F(f)+1$. We may use this fact throughout without mention.


The following lemma shows that there are two types of edges that determine the connectivity relation in $G\setminus F$ between the connected components of $T\setminus F$.

\begin{lemma}
\label{lemma:components_through_back-edges}
Let $e$ be an edge of $G\setminus F$ whose endpoints lie in different connected components of $T\setminus F$. Then $e$ is a back-edge and either $(i)$ both endpoints of $e$ lie in internal components, or $(ii)$ one endpoint of $e$ lies in a hanging subtree $H$, and the other endpoint lies in an internal component $C$ that is an ancestor of $H$.  
\end{lemma}
\begin{proof}
Let $e=(x,y)$, let $C$ be the connected component of $T\setminus F$ that contains $x$, and let $C'$ be the connected component of $T\setminus F$ that contains $y$. We have that $e$ cannot be a tree-edge, because otherwise $x$ and $y$ would be connected in $T\setminus F$. Thus, $e$ is a back-edge. Since $x,y$ are the endpoints of a back-edge, they are related as ancestor and descendant.  We may assume w.l.o.g. that $r_C>r_{C'}$. We will show that this implies that $x$ is a descendant of $y$. So let us suppose, for the sake of contradiction, that $x$ is an ancestor of $y$. Since $y\in C'$, we have that $y$ is a descendant of $r_{C'}$. Since $x$ is an ancestor of $y$ that does not lie in $C'$, we have that $x$ does not lie on the tree path $T[y,r_{C'}]$. Thus, $x$ is a proper ancestor of $r_{C'}$, and so $x<r_{C'}$. Since $x\in C$, we have that $x\geq r_C$. Thus, we have $r_C\leq x<r_{C'}$, which contradicts the assumption $r_C>r_{C'}$. Thus, we have shown that $x$ is a descendant of $y$. Now, since $y$ does not lie in $C$, we have that $y$ cannot lie on the tree path $T[x,r_C]$. Therefore, since $y$ is an ancestor of $x$, it must be a proper ancestor of $r_C$. And since $y\in C'$, we have that $y$ is a descendant of $r_{C'}$. Therefore, $r_C$ is a descendant of $r_{C'}$.

Thus we have shown that $C$ is a descendant of $C'$. This implies that $C'$ cannot be a hanging subtree of $T\setminus F$. To see this, suppose the contrary. Since $r_{C'}$ is a proper ancestor of $r_C$, we have that $r_{C'}$ is an ancestor of $p(r_C)$. ($p(r_C)$ is defined, precisely because $r_C$ has a proper ancestor, and therefore $r_C\neq r$.) Notice that $p(r_C)$ is a failed vertex (otherwise, $r_C$ would be connected with $p(r_C)$ through the parent edge $(r_C,p(r_C))$, contradicting the fact that $r_C$ is the root of a connected component of $T\setminus F$). But then we have that $r_C'$ is an ancestor of a failed vertex, contradicting our supposition that $C'$ is a hanging subtree of $T\setminus F$. We conclude that, among $C$ and $C'$, only $C$ can be a hanging subtree of $T\setminus F$.
\end{proof} 

\begin{corollary}
\label{corollary:connected_components}
Let $C,C'$ be two distinct connected components of $T\setminus F$ that are connected with an edge $e$ of $G\setminus F$. Assume w.l.o.g. that $r_{C'}<r_C$. Then $C'$ is an ancestor of $C$.
\end{corollary}
\begin{proof}
Lemma~\ref{lemma:components_through_back-edges} implies that $e$ is a back-edge. Let $e=(x,y)$, and assume w.l.o.g. that $x$ is a descendant of $y$. We know that either $x\in C$ and $y\in C'$, or reversely. Let us suppose, for the sake of contradiction, that $x\in C'$ (and thus $y\in C$). This implies that $x$ is a descendant of $r_{C'}$. Thus, $x$ is a common descendant of $r_{C'}$ and $y$. This implies that $r_{C'}$ and $y$ are related as ancestor and descendant. We have that $y$ cannot be a descendant of $r_{C'}$, because this would imply that $y\in T[r_{C'},x]$ (but $y$ lies outside of $C'$). Thus, we have that $y$ is a proper ancestor of $r_{C'}$, and therefore $y<r_{C'}$. Since $r_{C'}<r_C$, this implies that $y<r_C$, and therefore $y$ cannot be a descendant of $r_C$ -- contradicting the fact that $y\in C$.

Thus we have shown that $x\in C$ and $y\in C'$. $x\in C$ implies that $x$ is a descendant of $r_C$. Thus, $x$ is a common descendant of $r_C$ and $y$. This implies that $r_C$ and $y$ are related as ancestor and descendant. We have that $y$ cannot be a descendant of $r_C$, because this would imply $y\in T[r_C,x]$ (but $y$ lies outside of $C$). Thus, $r_C$ is a descendant of $y$. Also, $y\in C'$ implies that $y$ is a descendant of $r_{C'}$. Thus, we conclude that $r_C$ is a descendant of $r_{C'}$. 
\end{proof}

%
%
%

The following lemma provides an algorithmically useful criterion to determine whether a connected component of $T\setminus F$ -- a hanging subtree or an internal component -- is connected with an internal component of $T\setminus F$ through a back-edge.

\begin{lemma}
\label{lemma:ancestor_c_back-edge}
Let $C,C'$ be two connected components of $T\setminus F$ such that $C'$ is an internal component that is an ancestor of $C$, and let $b$ be the boundary vertex of $C'$ that is an ancestor of $C$. Then there is a back-edge from $C$ to $C'$ if and only if there is a back-edge from $C$ whose lower end lies in $[r_{C'},p(b)]$. 
\end{lemma}
\begin{proof}
First, let us explain the existence of $b$. Since $C'$ is an ancestor of $C$, we have that $r_{C'}$ is an ancestor of $r_C$. Therefore, Lemma~\ref{lemma:ic-info}$(2)$ implies that there is a unique boundary vertex $b$ of $C'$ that is an ancestor of $r_C$. Thus, $b$ is an ancestor of $C$.

($\Rightarrow$) Let $e=(x,y)$ be a back-edge from $C$ to $C'$, and assume w.l.o.g. that $x$ lies in $C$. Since $e$ is a back-edge, we have that either $x$ is a descendant of $y$, or reversely. Let us suppose, for the sake of contradiction, that $y$ is a descendant of $x$. Since $x\in C$, we have that $x\geq r_C$. Since $y$ is a descendant of $x$, we have that $y>x$. Thus, $y>r_C$. Since $(x,y)$ is a back-edge from $C$ to $C'$ and $x\in C$, we have that $y\in C'$. This implies that there must be a failed vertex on the tree path $T[y,x]$. (Otherwise, $y$ would be connected with $x$, and therefore $C'$ would be connected with $C$, which is absurd.) Let $f$ be the maximum failed vertex on the tree path $T[y,x]$. Then, the connected component of $T\setminus F$ that contains $y$ has a child of $f$ as a root. But this root is $r_{C'}$, and therefore we have $r_{C'}>f>x\geq r_C$, in contradiction to the assumption that $C'$ is an ancestor of $C$. Thus we have shown that $x$ is a descendant of $y$. Since $y$ is an ancestor of $x$ that does not lie in $C$, we have that $y$ does not lie on the tree path $T[x,r_C]$, and therefore it must be a proper ancestor of $r_C$. Thus, since $y\in C'$, we have that $y$ lies on the tree path $T[p(b),r_{C'}]$. This implies that $y\in [r_{C'},p(b)]$.

($\Leftarrow$) 
Let $e=(x,y)$ be a back-edge from $C$ whose lower end lies in $[r_{C'},p(b)]$. We may assume w.l.o.g. that $x\in C$. Thus, we have that $y\in[r_{C'},p(b)]$, and that $y$ is an ancestor of $x$. Since $x\in C$, we have that $x$ is a descendant of $r_C$. Since $b$ is an ancestor of $C$, we have that $b$ is an ancestor of $r_C$. Thus, $x$ is a descendant of $b$, and therefore a descendant of $p(b)$. This means that the tree path $T[p(b),r_{C'}]$ consists of ancestors of $x$. Thus, since $y$ is an ancestor of $x$ with $y\in [r_{C'},p(b)]$, we have that $y\in T[p(b),r_{C'}]$. Since $b$ is a boundary vertex of $C'$, we have that all vertices on the tree path $T[p(b),r_{C'}]$ lie in $C'$. In particular, we have $y\in C'$.
\end{proof}

\begin{definition}
\label{definition:R}
Let $\mathcal{R}$ be a multigraph where $V(\mathcal{R})$ is the set of the roots of the internal components of $T\setminus F$, and $E(\mathcal{R})$ satisfies the following three properties:

\begin{enumerate}[label={(\arabic*)}]
\item{For every back-edge connecting two internal components $C$ and $C'$, there is an edge $(r_C,r_{C'})$ in $\mathcal{R}$.}
\item{Let $H$ be a hanging subtree of a failed vertex $f$, and let $C_1,\dots,C_k$ be the internal components that are connected with $H$ through a back-edge. (By Lemma~\ref{lemma:components_through_back-edges}, all of $C_1,\dots,C_k$ are ancestors of $H$.) Assume w.l.o.g. that $C_k$ is an ancestor of all $C_1,\dots,C_{k-1}$. Then $\mathcal{R}$ contains the edges $(r_{C_1},r_{C_k}),(r_{C_2},r_{C_k}),\dots,(r_{C_{k-1}},r_{C_k})$.}
\item{Every edge of $\mathcal{R}$ is given by either $(1)$ or $(2)$, or it is an edge of the form $(r_C,r_{C'})$, where $C,C'$ are two internal components that are connected in $G\setminus F$.} 
\end{enumerate}

Then $\mathcal{R}$ is called a \emph{connectivity graph of the internal components} of $T\setminus F$. The edges of $(1)$ and $(2)$ are called Type-1 and Type-2, respectively.
\end{definition}

The following lemma shows that this graph captures the connectivity relationship of the internal components of $T\setminus F$ in $G\setminus F$. 

\begin{lemma}
\label{lemma:maintaining_components}
Let $\mathcal{R}$ be a connectivity graph of the internal components of $T\setminus F$. Then, two internal components $C,C'$ of $T\setminus F$ are connected in $G\setminus F$ if and only if $r_C,r_{C'}$ are connected in $\mathcal{R}$.
\end{lemma}
\begin{proof}
($\Rightarrow$) Let $C,C'$ be two internal components of $T\setminus F$ that are connected in $G\setminus F$. This means that there is a sequence $C_1,\dots,C_k$ of pairwise distinct connected components of $T\setminus F$, and a sequence of back-edges $e_1,\dots,e_{k-1}$, such that: $C_1=C$, $C_k=C'$, and $e_i$ connects $C_i$ and $C_{i+1}$, for every $i\in\{1,\dots,k-1\}$. By Lemma~\ref{lemma:components_through_back-edges}, we have that, for every $i\in\{1,\dots,k-1\}$, either $(1)$ $C_i$ and $C_{i+1}$ are internal components that are related as ancestor and descendant, or $(2)$ one of $C_i,C_{i+1}$ is a hanging subtree, and the other is an internal component that is an ancestor of it. 

Let $i$ be an index in $\{1,\dots,k-1\}$. If $(1)$ is true, then there is a Type-1 edge $(r_{C_i},r_{C_{i+1}})$ in $\mathcal{R}$. 
If $(2)$ is true, then one of $C_i,C_{i+1}$ is a hanging subtree. Let us assume that $C_i$ is a hanging subtree. Since there is a back-edge connecting $C_i$ with $C_{i+1}$, we may consider the lowest internal component $\tilde{C}$ that is an ancestor of $C_i$ and is connected with it through a back-edge. If $C_{i+1}=\tilde{C}$, then we imply nothing at this point. Otherwise, we have that $\mathcal{R}$ contains the Type-2 edge $(r_{C_{i+1}},r_{\tilde{C}})$. Now, since $C_1$ is an internal component, we have that $C_i\neq C_1$, and therefore $C_{i-1}$ is defined. By Lemma~\ref{lemma:components_through_back-edges}, we have that $C_{i-1}$ is also an internal component, that is connected with $C_i$ through a back-edge. Again, if $C_{i-1}=\tilde{C}$, then we imply nothing at this point. Otherwise, we have that $\mathcal{R}$ contains the Type-2 edge $(r_{C_{i-1}},r_{\tilde{C}})$. Thus, there are three possibilities to consider: either $C_{i-1}=\tilde{C}$ and $C_{i+1}\neq\tilde{C}$, or $C_{i-1}\neq\tilde{C}$ and $C_{i+1}=\tilde{C}$, or $C_{i-1}\neq\tilde{C}$ and $C_{i+1}\neq\tilde{C}$. In any case, we can see that $r_{C_{i-1}}$ is connected with $r_{C_{i+1}}$ in $\mathcal{R}$ -- either directly, or through $r_{\tilde{C}}$. Similarly, if we assume that $C_{i+1}$ is a hanging subtree (and $C_i$ is an internal component), then we have that $r_{C_i}$ is connected with $r_{C_{i+2}}$ in $\mathcal{R}$.

From all this we infer that, if $C_{i(1)},\dots,C_{i(t)}$ is the subsequence of $C_1,\dots,C_k$ that consists of the internal components, then $r_{C_{i(1)}},\dots,r_{C_{i(t)}}$ are connected in $\mathcal{R}$. In particular, we have that $r_C$ and $r_{C'}$ are connected in $\mathcal{R}$. 

($\Leftarrow$) Let $e=(r_C,r_{C'})$ be an edge of $\mathcal{R}$. If $e$ is a Type-1 edge, then there is a back-edge that connects $C$ and $C'$ in $G\setminus F$. Otherwise, there is a hanging subtree of $T\setminus F$ that is connected with both $C$ and $C'$ in $G\setminus F$ (through back-edges). In any case, we have that $C,C'$ are connected in $G\setminus F$. Since this is true for any edge of $\mathcal{R}$, we conclude that, if $r_C,r_{C'}$ are two vertices connected in $\mathcal{R}$, then $C,C'$ are connected in $G\setminus F$.
\end{proof}

\subsection{Handling the updates: construction of a connectivity graph for the internal components of $T\setminus F$}
\label{section:update}

Given a set of failed vertices $F$, with $|F|=d\leq d_{\star}$, we will show how we can construct a connectivity graph $\mathcal{R}$ for the internal components of $T\setminus F$, using $O(d^4)$ calls to 2D-range-emptiness queries. Recall that $V(\mathcal{R})$ is the set of the roots of the internal components of $T\setminus F$.
  
Algorithm~\ref{algorithm:type-1-edges} shows how we can find all Type-1 edges of $\mathcal{R}$. The idea is basically to perform 2D-range-emptiness queries for every pair of internal components, in order to determine the existence of a back-edge that connects them. More precisely, we work as follows. Let $C$ be an internal component of $T\setminus F$. Then it is sufficient to check every ancestor component $C'$ of $C$, in order to determine whether there is a back-edge from $C$ to $C'$ (see Corollary~\ref{corollary:connected_components}). Let $f_1,\dots,f_k$ be the boundary vertices of $C$, sorted in increasing order. Let also $f'$ be the boundary vertex of $C'$ that is an ancestor of $C$, and let $I=[r_{C'},p(f')]$. Then we perform 2D-range-emptiness queries for the existence of a back-edge on the rectangles $[r_C,f_1-1]\times I,[f_1+\mathit{ND}(f_1),f_2-1]\times I,\dots,[f_k+\mathit{ND}(f_k),r_C+\mathit{ND}(r_C)-1]\times I$. We know that there is a back-edge connecting $C$ and $C'$ if and only if at least one of those queries is positive (see Lemma~\ref{lemma:ic-info}$(3)$ and Lemma~\ref{lemma:ancestor_c_back-edge}). If that is the case, then we add the edge $(r_C,r_{C'})$ to $\mathcal{R}$.

Observe that the total number of 2D-range-emptiness queries that we perform is $O(d^2)$, because every one of them corresponds to a triple $(C,f,C')$, where $C,C'$ are internal components, $C'$ is an ancestor of $C$, and $f$ is a boundary vertex of $C$, or $r_C$. And if $C_1,\dots,C_k$ are all the internal components of $T\setminus F$, then the number of those triples is bounded by $(|\partial(C_1)|+1)\cdot d+\dots+(|\partial(C_k)|+1)\cdot d = (|\partial(C_1)|+\dots+|\partial(C_k)|+k)\cdot d \leq (d+k)\cdot d \leq (d+d)\cdot d = O(d^2)$.

\begin{algorithm}[!h]
\caption{\textsf{Compute all Type-1 edges to construct a connectivity graph $\mathcal{R}$ for the internal components of $T\setminus F$}}
\label{algorithm:type-1-edges}
\LinesNumbered
\DontPrintSemicolon
\ForEach{internal component $C$ of $T\setminus F$}{
\label{line:iterate_c}
  let $f_1,\dots,f_k$ be the boundary vertices of $C$, sorted in increasing order\;
  \label{line:boundary_c}
\textcolor{red}{\tcp{process every internal component $C'$ that is an ancestor of $C$}}
  set $f'\leftarrow p(r_C)$\;
  \While{$f'\neq\bot$}{
  \label{line:while_ancestors}
    \If{$p(f')\neq\mathit{parent}_F(f')$}{
    \label{line:apply_parent_criterion}
    let $C'$ be the internal component of $T\setminus F$ with $f'\in\partial(C')$\;
    set $I\leftarrow [r_{C'},p(f')]$\;
    \If{at least one of the following queries is positive:\\ $\mathtt{2D\_range}([r_C,f_1-1]\times I)$\\ $\mathtt{2D\_range}([f_1+\mathit{ND}(f_1),f_2-1]\times I)$\\$\dots$\\$\mathtt{2D\_range}([f_{k-1}+\mathit{ND}(f_{k-1}),f_k-1]\times I)$\\$\mathtt{2D\_range}([f_k+\mathit{ND}(f_k),r_C+\mathit{ND}(r_C)-1]\times I)$}{
      add the Type-1 edge $(r_C,r_{C'})$ to $\mathcal{R}$\;
    }}
    $f'\leftarrow\mathit{parent}_F(f')$\;
  }
}
\end{algorithm}

\begin{proposition}
Algorithm~\ref{algorithm:type-1-edges} correctly computes all Type-1 edges to construct a connectivity graph for the internal components of $T\setminus F$. The running time of this algorithm is $O(d^2\log n)$.
\end{proposition}
\begin{proof}
First, we need to provide a method to efficiently iterate over the collection of the internal components and their boundary vertices (Lines~\ref{line:iterate_c} and \ref{line:boundary_c}), and then we have to prove that the \textbf{while} loop in Line~\ref{line:while_ancestors} is sufficient to access all internal components that are ancestors of $C$. Then, the correctness and the $O(d^2\log n)$ time-bound follow from the analysis above (in the main text).

Every internal component $C$ of $T\setminus F$ is determined by its root $r_C$. By Lemma~\ref{lemma:ic-info}$(1)$, we have that either $r_C=r$, or $p(r_C)$ is a failed vertex. If $r_C=r$ then $C$ has no ancestor internal components, and therefore we may ignore this case. So let $p(r_C)=f$ be a failed vertex. Then, by Lemma~\ref{lemma:f_is_boundary}, the boundary vertices of $C$ are given by the children of $f$ in the F-forest that are descendants (in $T$) of $r_C$. 

Thus, we may work as follows. First, we sort the children of every failed vertex $f$ in the F-forest in increasing order. This takes $O(d\log d)$ time in total. Then, for every failed vertex $f$, we traverse its list of children $L$ (in the F-forest) in order. For every maximal segment $S$ of $L$ that consists of descendants of the same child $c$ of $f$ in $T$, we know that either $c$ is the root of an internal component with boundary $S$, or $c\in F$. (For every $f'\in L$ that we meet, we can use a level-ancestor query to find in constant time the child of $f$ in $T$ that is an ancestor of $f'$.) Thus, Lines~\ref{line:iterate_c} and \ref{line:boundary_c} need $O(d\log d)$ time in total.

Now let $C$ be an internal component with $r_C\neq r$. Then we have that $p(r_C)$ is a failed vertex. Now let $C'$ be an internal component that is a proper ancestor of $C$. This means that $r_{C'}$ is a proper ancestor of $r_C$, and therefore $r_{C'}$ is an ancestor of $f=p(r_C)$. Then, by Lemma~\ref{lemma:ic-info}$(2)$ we have that there is a boundary vertex $b$ of $C'$ that is an ancestor of $f$ (in $T$). Since the set of failed vertices that are ancestors of $f$ (in $T$) coincide with the set of ancestors of $f$ in the F-forest, we have that the \textbf{while} loop in Line~\ref{line:while_ancestors} will eventually reach $b$. Then we can retrieve $C'$ (more precisely: $r_{C'}$) in constant time using Lemma~\ref{lemma:f_is_boundary}. The purpose of Line~\ref{line:apply_parent_criterion} is to apply the criterion of Lemma~\ref{lemma:f_is_boundary}, in order to check whether $f'$ is a boundary vertex of an internal component.
\end{proof}

The construction of Type-2 edges is not so straightforward. For every failed vertex $f$, and every two internal components $C$ and $C'$, such that $C$ is an ancestor of $f$ and $C'$ is an ancestor of $C$, we would like to know whether there is a hanging subtree of $f$, from which stem a back-edge $e$ with an endpoint in $C$ and a back-edge $e'$ with an endpoint in $C'$. The straightforward way to determine this is the following. Let $b$ (resp., $b'$) be the boundary vertex of $C$ (resp., $C'$) that is an ancestor of $f$. Then, for every hanging subtree of $f$ with root $c$, we perform 2D-range-emptiness queries on the rectangles $[c,c+\mathit{ND}(c)-1]\times[r_C,p(b)]$ and $[c,c+\mathit{ND}(c)-1]\times[r_{C'},p(b')]$. If both queries are positive, then we know that $C$ and $C'$ are connected in $G\setminus F$ through the hanging subtree with root $c$. 

Obviously, this method is not efficient in general, because the number of hanging subtrees of $f$ can be very close to $n$. However, it is the basis for our more efficient method. The idea is to perform a lot of those queries at once, for large batches of hanging subtrees. More specifically, we perform the queries on \emph{consecutive} hanging subtrees of $f$ (i.e., their roots are consecutive children of $f$), for which we know that the answer is positive on $C'$ (i.e., for every one of those subtrees, there certainly exists a back-edge that connects it with $C'$). In order for this idea to work, we have to rearrange properly the lists of children of all vertices. (Otherwise, the hanging subtrees of $f$ that are connected with $C'$ through a back-edge may not be consecutive in the list of children of $f$.) In effect, we maintain several DFS trees (specifically: $d_{\star}$), and several 2D-range-emptiness data structures, one for every different arrangement of the children lists.

Let us elaborate on this idea.
Let $H$ be a hanging subtree of $f$ that connects some internal components, and let $C'$ be the lowest one among them (i.e., the one that is an ancestor of all the others). Then we have that the lower ends of all back-edges that stem from $H$ and end in ancestors of $C'$ are failed vertices that are ancestors of $C'$. Thus, since there are at most $d$ failed vertices in total, we have that at least one among $\mathit{low}_1(r_H),\dots,\mathit{low}_d(r_H)$ is in $C'$. 
In other words, $r_H$ is one of the children of $f$ whose $\mathit{low}_i$ point is in $C'$, for some $i\in\{1,\dots,d\}$.
Now, assume that for every $i\in\{1,\dots,d_{\star}\}$, we have a copy of the list of the children of $f$ sorted in increasing order w.r.t. the $\mathit{low}_i$ point; let us call this list $L_i(f)$, and let it be stored in way that allows for binary search w.r.t. the $\mathit{low}_i$ point. Then, for every internal component $C$ that is an ancestor of $f$, we can find the segment $S_i(C)$ of $L_i(f)$ that consists of the children of $f$ whose $\mathit{low}_i$ point lies in $C$, by searching for the leftmost and the righmost child in $L_i(f)$ whose $\mathit{low}_i$ point lies in $[r_C,p(b)]$, where $b$ is the boundary vertex of $C$ that is an ancestor of $f$.
%
%

Now let $i\in\{1,\dots,d\}$ be such that $\mathit{low}_i(r_H)\in C'$. Then we have that $r_H\in S_i(C')$. Furthermore, we have that every child of $f$ that lies in $S_i(C')$ and is the root of a hanging subtree $H'$ of $f$ has the property that $H'$ is also connected with $C'$ through a back-edge. Thus,  we would like to be able to perform 2D-range-emptiness queries as above on the subset $S$ of $S_i(C')$ that consists of roots of hanging subtrees, in order to determine the connectivity (in $G\setminus F$) of $C'$ with all internal components $C$ that are ancestors of $f$ and descendants of $C'$. We could do this efficiently if we had the guarantee that $S$ consists of large segments of consecutive children of $f$. We can accommodate for that during the preprocessing phase: for every $i\in\{1,\dots,d_{\star}\}$, we perform a DFS of $T$, starting from $r$, where the adjacency list of every vertex $v$ is given by $L_i(v)$.\footnote{I.e., it is necessary that the vertices in the adjacency list of $v$ appear in the same order as in $L_i(v)$.} 
Let $T_i$ be the resulting DFS tree, and let $\mathit{DFS}_i$ be the corresponding DFS numbering. Then, with the DFS numbering of $T_i$, we initialize a data structure $\mathit{2D\_range\_i}$, for answering 2D-range-emptiness queries for back-edges w.r.t. $T_i$ in subrectangles of $[1,n]\times[1,n]$. 

Now let us see how everything is put together. Let $H$ be a hanging subtree of $f$ that connects two internal components $C_1$ and $C_2$, and let $b_1$ and $b_2$ be the boundary vertices of $C_1$ and $C_2$, respectively, that are ancestors of $f$. Let $C'$ be the lowest internal component that is connected through a back-edge with $H$. Then there is an $i\in\{1,\dots,d\}$ such that $\mathit{low}_i(r_H)\in C'$. Let $S$ be the maximal segment of $S_i(C')$ that contains $r_H$ and consists of roots of hanging subtrees, let $L$ be the minimum of $S$ and let $R$ be the maximum of $S$.\footnote{Notice that, due to the construction of $T_i$, we have that $\mathit{DFS}_i(L)$ and $\mathit{DFS}_i(R)$ are also the minimum and the maximum, respectively, of $\mathit{DFS}_i(S)$.} Then the 2D-range-emptiness queries on $[L,R+\mathit{ND}(R)-1]_i\times[r_{C_1},p(b_1)]_i$ and $[L,R+\mathit{ND}(R)-1]_i\times[r_{C_2},p(b_2)]_i$ with $\mathit{2D\_range\_i}$ are both positive, and so we will add the edges $(r_{C_1},r_{C'})$ and $(r_{C_2},r_{C'})$ in $\mathcal{R}$. This will maintain in $\mathcal{R}$ the information that $C'$, $C_1$ and $C_2$, are connected with the same hanging subtree of $f$.  


The algorithm that constructs enough Type-2 edges to make $\mathcal{R}$ a connectivity graph of the internal components of $T\setminus F$ is given in Algorithm~\ref{algorithm:type-2-edges}. The proof of correctness and time complexity is given in Proposition~\ref{proposition:algorithm-2-correctness}.

\begin{algorithm}[!h]
\caption{\textsf{Compute enough Type-2 edges to construct a connectivity graph for the internal components of $T\setminus F$}}
\label{algorithm:type-2-edges}
\LinesNumbered
\DontPrintSemicolon
\ForEach{failed vertex $f$}{
  \label{alg:line_for_f}
  \textcolor{cred}{\tcp{process all pairs of internal components that are ancestors of $f$}}
  set $f'\leftarrow \mathit{parent}_F(f)$\;  
  \While{$f'\neq\bot$}{
    let $C'$ be the internal component with $f'\in\partial(C')$\;
    \textcolor{cblue}{\tcp{skip the following if $C'$ does not exist, and go immediately to Line~{$26$}}}
    \ForEach{$i\in\{1,\dots,d\}$}{
      \label{alg:line_for_S}
      let $\mathcal{S}_i$ be the collection of all maximal segments of $L_i(f)$ that consist of roots of hanging subtrees with their $\mathit{low_i}$ point in $C'$\;
      \label{alg:line_maximal_segments}
    }
    \textcolor{cred}{\tcp{process all internal components $C$ that are ancestors of $f$ and descendants of $C'$}}
    set $f''\leftarrow f$\;
    \While{$f''\neq f'$}{
    \label{line:while_boundary}
      let $C$ be the internal component with $f''\in\partial(C)$\;
      \textcolor{cblue}{\tcp{skip the following if $C$ does not exist, and go immediately to Line~{$24$}}}
      \textcolor{cred}{\tcp{check if $C$ is connected with $C'$ through at least one hanging subtree of $f$}}
      \ForEach{$i\in\{1,\dots,d\}$}{
        \ForEach{$S\in\mathcal{S}_i$}{
        \label{line:S_for}
          let $L\leftarrow\mathit{min}(S)$ and $R\leftarrow\mathit{max}(S)$\;
          \If{$\mathtt{2D\_range\_i}([L,R+\mathit{ND}(R)-1]_i\times [r_{C},p(f'')]_i)=\mathit{true}$}{
          \label{line:2d_query}
            add the Type-2 edge $(r_C,r_{C'})$ to $\mathcal{R}$\;
          }
        }
      }
      $f''\leftarrow\mathit{parent}_F(f'')$\;
      \label{alg:line_next_f}
    }
    $f'\leftarrow\mathit{parent}_F(f')$\;
    \label{alg:line_next_f'}
  }
}
\end{algorithm}

\begin{proposition}
\label{proposition:algorithm-2-correctness}
Algorithm~\ref{algorithm:type-2-edges} computes enough Type-2 edges to construct a connectivity graph $\mathcal{R}$ for the internal components of $T\setminus F$ (supposing that $\mathcal{R}$ contains all Type-1 edges). The running time of this algorithm is $O(d^4\log n)$.
\end{proposition}
\begin{proof}
By definition, it is sufficient to prove the following: for every failed vertex $f$, and every hanging subtree $H$ of $f$, let $C'$ be the lowest internal component that is connected with $H$ through a back-edge; then, for every internal component $C\neq C'$ that is connected with $H$ through a back-edge, there is an edge $(r_C,r_{C'})$ added to $\mathcal{R}$. And conversely: that these are all the Type-2 edges that are added to $\mathcal{R}$, and that any other edge $(r_C,r_{C'})$ that is added to $\mathcal{R}$ with Algorithm~\ref{algorithm:type-2-edges} has the property that $C$ and $C'$ are connected with the same hanging subtree through back-edges.

So let $f$ be a failed vertex, let $H$ be a hanging subtree of $f$, and let $C'$ be the lowest internal component that is connected with $H$ through a back-edge. Let us assume that $f\notin\partial(C')$. (Otherwise, there is no internal component $C$ that is an ancestor of $H$ and a descendant of $C'$, and therefore $H$ does not induce any Type-2 edges.) 
Let $C\neq C'$ be an internal component that is connected with $H$ through a back-edge. Since $C'$ is the lowest internal component that is connected with $H$ through a back-edge, by the analysis above (in the main text) we have that there is an $i\in\{1,\dots,d\}$ such that $\mathit{low}_i(r_H)\in C'$. Thus, we may consider the maximal segment $S$ of $L_i(f)$ that contains $r_H$ and consists of roots of hanging subtrees whose $\mathit{low}_i$ point is in $C'$. Let $L$ and $R$ be the minimum and the maximum, respectively, of $S$. By construction of $T_i$, we have that $S$ is sorted in increasing order w.r.t. the $\mathit{DFS}_i$ numbering. Thus, the interval $[L,R+\mathit{ND}(R)-1]_i$ consists of the descendants of the vertices in $S$ in $T_i$. Since the vertices in $T_i$ have the same ancestry relation as in $T$, we have that the set $\mathit{DFS}_i(S)$ consists of children of $f$ in $T_i$ that are roots of hanging subtrees with $\mathit{low}_i$ in $C'$. 

Now let $f''$ be the boundary vertex of $C$ that is an ancestor of $f$.  Then we also have that $f''$ is the boundary vertex of $C$ that is an ancestor of $r_H$ (since $f=p(r_H)$). Thus, Lemma~\ref{lemma:ancestor_c_back-edge} implies that there is a back-edge from $H$ to $[r_C,p(f'')]$. Therefore, there is also a back-edge (w.r.t. $T_i$) from $\mathit{DFS}_i(H)$ to $[\mathit{DFS}_i(r_C),\mathit{DFS}_i(p(f''))]$. This implies that the 2D-range query in Line~\ref{line:2d_query} is true, and therefore the Type-2 edge $(r_C,r_{C'})$ will be correctly added to $\mathcal{R}$. It is not difficult to see that the converse is also true: whenever the 2D-range query in Line~\ref{line:2d_query} is true, we can be certain that there is a hanging subtree of $f$ that is connected through a back-edge with both $C'$ and $C$.

Let us analyze the running time of Algorithm~\ref{algorithm:type-2-edges}. 
First, we will provide a method to implement Line~\ref{alg:line_maximal_segments}, i.e., how to find, for every internal component $C'$ that is an ancestor of $f$, and every $i\in\{1,\dots,d\}$, the collection $\mathcal{S}_i$ of the maximal segments of $L_i(f)$ that consist of roots of hanging subtrees whose $\mathit{low}_i$ points lie in $C'$. There are many ways to do this, but for the sake of simplicity we will provide a relatively straightforward method that incurs total cost $O(d^3+d^2\log n)$. The idea is to collect the children of $f$, at the beginning of the \textbf{for} loop in Line~\ref{alg:line_for_f}, that are ancestors of failed vertices. To do this, we collect the failed vertices $f_1,\dots,f_k$ that are children of $f$ in the F-forest, and then we perform a level-ancestor query (in $T$) for every $f_i$ to find the child $c_i$ of $f$ that is an ancestor of $f_i$. 
Then we keep $d$ copies, $\mathcal{C}_1,\dots,\mathcal{C}_d$, of the collection $\{c_1,\dots,c_k\}$. (We note that some $c_i,c_j$, for $i,j\in\{1,\dots,k\}$ with $i\neq j$, may coincide. We ignore those repetitions.) For every $i\in\{1,\dots,d\}$, we let $\mathcal{C}_i$ be sorted in increasing order w.r.t. the $\mathit{low}_i$ points. Now, at the beginning of the \textbf{for} loop in Line~\ref{alg:line_for_S}, we can use binary search to find in $O(\log n)$ time the (endpoints of the) segment $S$ of $L_i(f)$ that consists of all children of $f$ with their $\mathit{low}_i$ point in $C'$. Let $L$ and $R$ be the minimum and the maximum, respectively, of $S$. 
Then we traverse the list $\mathcal{C}_i$, and, for every $c\in\mathcal{C}_i$ that we meet, we check whether $c$ is in $S$. (This is done by simply checking whether $L\leq c\leq R$.) If that is the case, then we collect the (possibly empty) subsegment $[L,c-1]$ (i.e., the pair of its endpoints), and we remove $[L,c]$ from $S$ (i.e., we set $L\leftarrow c+1$). We repeat this process while traversing $\mathcal{C}_i$ until we reach its end, and finally we collect the (possibly empty) remainder of $S$ (i.e., the pair of its endpoints). The collection of the non-empty subsegments we have gathered is precisely $\mathcal{S}_i$.

The cost of this method is as follows. For every failed vertex $f$, we need time analogous to its number of children in the F-forest to create the collection $\{c_1,\dots,c_k\}$. Then we make $d$ copies of this collection, and we perform a sorting in every one of them. This takes time $O(d\cdot k\log k)$, where $k$ is the number of children of $f$ in the F-forest. Since this is performed for every failed vertex, it incurs total cost $O(d^2\log d)$. Now, for every failed vertex $f$, and every internal component $C'$ that is an ancestor of $f$, we need $O(\log n)$ time to find the segment $S$, as described above. This is how we get an additional $O(d^2\log n)$ cost in total. Now, for this $f$ and $C'$, and for every $i\in\{1,\dots,d\}$, we have to traverse the list $\mathcal{C}_i$ as above (while performing operations that take constant time). Since the size of $\mathcal{C}_i$ equals the number of children of $f$ in the F-forest, this incurs cost $O(d^3)$ in total.

By the analysis above, we have that the total cost of Line~\ref{alg:line_maximal_segments} is $O(d^3+d^2\log n)$.\footnote{In RAM machines with $O(\log n)$ word size, we can use van Emde Boas trees in order to perform the binary searches above as predecessor/successor queries, and so we can reduce the ``$\log n$'' factor to ``$\log\log n$''.} Thus, it remains to upper bound the times that the 2D-range queries in Line~\ref{line:2d_query} are performed. To do this, we introduce the following notation. Let $f$ be a failed vertex, and let $\mathcal{C}$ denote the collection of the internal components that are ancestors of $f$. Then, for every $C'\in\mathcal{C}$, and every $i\in\{1,\dots,d\}$, we let $\mathcal{S}_i(C')$ denote the collection of the maximal segments of $L_i(f)$ that consist of roots of hanging subtrees of $f$ whose $\mathit{low}_i$ point lies in $C'$. Then we can see that the number of 2D-range queries in Line~\ref{line:2d_query} during the processing of $f$ (during the outer \textbf{for} loop in Line~\ref{alg:line_for_f}) is precisely $\sum_{C'\in\mathcal{C}}\sum_{C}\sum_{i\in\{1,\dots,d\}}|\mathcal{S}_i(C')|$ $(*)$, where the second sum is indexed over the internal components $C$ that are ancestors of $f$ and descendants of $C'$. 

Now fix an $i\in\{1,\dots,d\}$. For every $C'\in\mathcal{C}$, let $S_i(C')$ denote the the maximal segment of $L_i(f)$ that consists of children of $f$ whose $\mathit{low}_i$ point lies in $C'$. Notice that every segment in $\mathcal{S}_i(C')$ is contained entirely within $S_i(C')$. Since the internal components in $\mathcal{C}$ are pairwise disjoint, we have that the segments in $\{S_i(C')\mid C'\in\mathcal{C}\}$ are pairwise disjoint, and therefore their total number is bounded by $|\mathcal{C}|\leq d$. Since the number of failed vertices is $d$, the number of childen of $f$ that are ancestors of failed vertices is at most $d$. It is precisely the existence of those children that may force the segments in $\{S_i(C')\mid C'\in\mathcal{C}\}$ to be partitioned further in order to get $\bigcup\{\mathcal{S}_i(C')\mid C'\in\mathcal{C}\}$. But every such child breaks the segment $S_i(C')$, for a $C'\in\mathcal{C}$, into at most two subsegments. (Recall the analysis above that concerns the implementation of Line~\ref{alg:line_maximal_segments}.) Thus, the segments in $\{S_i(C')\mid C'\in\mathcal{C}\}$ must be partitioned at most $d$ times in order to get $\bigcup\{\mathcal{S}_i(C')\mid C'\in\mathcal{C}\}$. Thus, we have $\sum_{C'\in\mathcal{C}}|\mathcal{S}_i(C')|\leq d+d=O(d)$. This implies that $\sum_{i\in\{1,\dots,d\}}\sum_{C'\in\mathcal{C}}|\mathcal{S}_i(C')|= O(d^2)$, and therefore the expression $(*)$ can be bounded by $O(d^3)$. Since this is true for every failed vertex $f$, we can bound the number of the 2D-range queries in Line~\ref{line:2d_query} by $O(d^4)$.

\end{proof}

\subsection{Answering the queries}
\label{section:query}
Assume that we have constructed a connectivity graph $\mathcal{R}$ for the internal components of $T\setminus F$, and that we have computed its connected components. Thus, given two internal components $C$ and $C'$, we can determine in constant time whether $C$ and $C'$ are connected in $G\setminus F$, by simply checking whether $r_C$ and $r_{C'}$ are in the same connected component of $\mathcal{R}$ (see Lemma~\ref{lemma:maintaining_components}).  

Now let $x,y$ be two vertices in $V(G)\setminus F$. In order to determine whether $x,y$ are connected in $G\setminus F$, we try to 
substitute $x,y$ with roots of internal components of $T\setminus F$, and then we reduce the query to those roots. Specifically, if $x$ (resp., $y$) belongs to an internal component $C$ of $T\setminus F$, then the connectivity between $x,y$ is the same as that between $r_C,y$ (resp., $x,r_C$). Otherwise, if $x$ (resp., $y$) belongs to a hanging subtree $H$ of $T\setminus F$, then we try to find an internal component that is connected with $H$ through a back-edge. If such an internal component $C$ exists, then we can substitute $x$ (resp., $y$) with $r_C$. Otherwise, $x,y$ are connected in $G\setminus F$ if and only if they belong to the same hanging subtree of $T\setminus F$. This idea is shown in Algorithm~\ref{algorithm:query}.

\begin{algorithm}[!h]
\caption{$\mathtt{query}(x,y)$}
\label{algorithm:query}
\LinesNumbered
\DontPrintSemicolon
\If{$x$ lies in an internal component $C$ and $y$ lies in an internal component $C'$}{
  \lIf{$r_C$ is connected with $r_{C'}$ in $\mathcal{R}$}{\textbf{return true}}
  \textbf{return false}
}
\textcolor{cblue}{\tcp{at least one of $x,y$ lies in a hanging subtree}}
\If{$x$ lies in a hanging subtree $H$}{
\label{line:x_in_H}
  \textcolor{red}{\tcp{check whether $H$ is connected with an internal component through a back-edge}}
  \For{$i\in\{1,\dots,d\}$}{
  \label{line:for_low}
    \If{$\mathit{low}_i(r_H)\neq\bot$ \textbf{and} $\mathit{low}_i(r_H)\notin F$}{
      \textbf{return} $\mathtt{query}(\mathit{low}_i(r_H),y)$\;
    }
  }
  \textcolor{cblue}{\tcp{there is no internal component that is connected with $H$ in $G\setminus F$}} 
  \lIf{$y$ lies in $H$}{\textbf{return true}}
  \textbf{return false}
}
\textbf{return} $\mathtt{query}(y,x)$\;
\end{algorithm}

\begin{proposition}
\label{proposition:query}
Given two vertices $x,y$ in $V(G)\setminus F$, Algorithm~\ref{algorithm:query} correctly determines whether $x,y$ are connected in $G\setminus F$. The running time of Algorithm~\ref{algorithm:query} is $O(d)$.
\end{proposition}
\begin{proof}
To prove correctness, we only have to deal with the case that $x$ lies in a hanging subtree $H$ of $T\setminus F$ (Line~\ref{line:x_in_H}). In this case, we simply have to check whether there is an edge in $G\setminus F$ that connects $H$  with another connected component of $T\setminus F$. Thus, according to Lemma~\ref{lemma:components_through_back-edges}, we have to check whether $H$ is connected with an internal component of $T\setminus F$ through a back-edge. If such an internal component exists, let $C$ be the lowest among them. Since $p(r_H)$ is a failed vertex and the number of failed vertices that are ancestors of $H$ is bounded by $d$, we have that there is at least one $i\in\{1,\dots,d\}$ such that $\mathit{low}_i(r_H)$ lies in $C$. Thus, the connectivity query for $x,y$ in $G\setminus F$ is equivalent to that for $\mathit{low}_i(r_H),y$. Since $\mathit{low}_i(r_H)$ belongs to the internal component $C$, eventually the algorithm will terminate, and it will produce the correct result. 

Otherwise, if there is no internal component that is connected with $H$ through a back-edge, then $H$ is a connected component of $G\setminus F$, and so $y$ is connected with $x$ in $G\setminus F$ if and only if $y$ also lies within $H$.
Now, if the \textbf{for} loop in Line~\ref{line:for_low} has exhausted the search and either $\mathit{low}_d(r_H)$ does not exist, or $\mathit{low}_d(r_H)$ is a failed vertex, then we can be certain that there is no back-edge that connects $H$ with the rest of the graph $G\setminus F$. 

Now we will establish the $O(d)$ time-bound. First, given a vertex $x\notin F$, we can determine the connected component of $T\setminus F$ that contains $x$ by finding the nearest failed vertex $f$ that is an ancestor of $x$. This is done in $O(d)$ time by finding the maximum failed vertex that is an ancestor of $x$. If no such vertex exists, then $x$ belongs to the internal component with root $r$. Otherwise, the root of the connected component $C$ of $T\setminus F$ that contains $x$ is given by the child of $f$ that is an ancestor of $x$. This child is determined in constant time with a level-ancestor query for the ancestor of $x$ whose depth equals $\mathit{depth}(f)+1$. Now, given the root $r_C$ of a connected component $C$ of $T\setminus F$, we can determine in $O(d)$ time whether $C$ is an internal component or a hanging subtree of $T\setminus F$ by checking whether there is a failed vertex that is a descendant of $r_C$. Finally, we can perform the search in Line~\ref{line:for_low} in $O(d)$ time, if we have the list of failed vertices sorted in increasing order. (We can have this done with an extra cost of $O(d\log d)$ during the update phase.) Then, we can easily check in $O(d)$ time whether there is an $i\in\{1,\dots,d\}$ such that $\mathit{low}_i(r_C)\in F$, because the list $\mathit{low}_1(r_C),\dots,\mathit{low}_d(r_C)$ is also sorted in increasing order. 
\end{proof}

\paragraph{Acknowledgements.} I would like to thank my advisor, Loukas Georgiadis, for helpful comments on this manuscript.

\bibliography{lipics-v2021-sample-article}

\end{document}